\documentclass[1p]{elsarticle}
\usepackage{amsmath,amssymb,amsthm,mathtools}
\usepackage{booktabs,array}
\usepackage{microtype}
\usepackage[hidelinks]{hyperref}
\usepackage{enumitem}

\newtheorem{theorem}{Theorem}[section]
\newtheorem{lemma}[theorem]{Lemma}
\newtheorem{corollary}[theorem]{Corollary}
\theoremstyle{definition}
\newtheorem{definition}[theorem]{Definition}
\newtheorem{remark}[theorem]{Remark}

\newcommand{\N}{\mathbb N}
\newcommand{\SR}{\mathsf{SR}}
\newcommand{\BSR}{\mathsf{BSR}}
\newcommand{\D}{\mathsf D}
\newcommand{\Ex}{\exists}
\newcommand{\All}{\forall}
\newcommand{\Alt}{\mathsf A}
\newcommand{\Cnt}{\#}
\newcommand{\DTIME}{\mathsf{DTIME}}
\newcommand{\NTIME}{\mathsf{NTIME}}
\newcommand{\ATIME}{\mathsf{ATIME}}
\newcommand{\DSPACE}{\mathsf{DSPACE}}
\newcommand{\NSPACE}{\mathsf{NSPACE}}
\newcommand{\ASPACE}{\mathsf{ASPACE}}
\newcommand{\co}{\mathsf{co}}
\newcommand{\Pclass}{\mathsf P}
\newcommand{\NP}{\mathsf{NP}}
\newcommand{\PSPACE}{\mathsf{PSPACE}}
\newcommand{\EXP}{\mathsf{EXP}}
\newcommand{\NEXP}{\mathsf{NEXP}}
\newcommand{\EXPSPACE}{\mathsf{EXPSPACE}}
\newcommand{\Lclass}{\mathsf L}
\newcommand{\NL}{\mathsf{NL}}
\newcommand{\Eclass}{\mathsf E}
\newcommand{\NE}{\mathsf{NE}}
\newcommand{\ESPACE}{\mathsf{ESPACE}}

\newcommand{\TISP}{\mathsf{TISP}}
\newcommand{\dotminus}{\mathbin{\dot{-}}}

\journal{Information and Computation}

\begin{document}
\begin{frontmatter}

\title{Step Recursion: Resource Profiles and Descent Quotients}
\author{Kirill Osipov}

\begin{abstract}
We develop a resource representation for step recursion in which mutable-state width and recursion descent are explicit and independent parameters. A width bound $u$ controls the size of the encoded machine state, while an effective descent $\rho$ determines the available recursion depth $\delta_\rho(u)$. For generalized-inverse descents, we derive the depth directly from generator growth and characterize the increasing sequences that can occur as generator orbits.

We then connect this depth--width geometry to standard finite-branching computation. Every deterministic bounded-state dynamics is realizable by a single ordinary bounded step recursion over a fixed finite numerical basis. Using deterministic, existential, universal, or alternating aggregation on the same local dynamics yields the corresponding machine semantics. After closure under the width reparameterizations needed to absorb fixed local cost, the resulting language classes are exactly the machine time--space classes on profiles $(\delta_\rho(u),u)$.

Finally, profile domination quotients effective descents by admissible width reparameterization.  Some depth curves collapse, yet polynomial widths support an explicit infinite strict hierarchy between the canonical polynomial- and exponential-depth profiles.  Thus descent remains a nonredundant resource coordinate after polynomial width reparameterization; standard complexity classes are calibration points.
\end{abstract}


\end{frontmatter}

\section{Introduction}

Recursion-theoretic descriptions of complexity traditionally obtain machine bounds from restrictions on recursion, composition, or term formation.  Classical examples range from Cobham and Thompson to Wagner and Bellantoni--Cook, with later recursion schemes for nondeterminism, alternation, and counting \cite{Cobham1965,Thompson1972,Wagner1979,BellantoniCook1992,Oitavem2013,Oitavem2022,VollmerWagner1996,DalLagoKahleOitavem2022}.  Wagner-style function algebras can already parameterize general time--space bounds; see Clote's survey \cite{Clote1999}.  Implicit computational complexity instead derives resource bounds from structural restrictions rather than explicit annotations \cite{DalLago2022}.  The present paper asks a different representation question: can finite-branching resource geometry be exposed directly through state width and recursion descent while the bounded-state schema and finite local transition basis remain fixed?

The distinction is visible already for a binary clock of width $u$.  Under
\[
 y\mapsto\lfloor y/2^l\rfloor
 \qquad\text{and}\qquad
 y\mapsto y\dotminus l,
\]
the same $u$-bit clock supports respectively $\Theta(u)$ and $\Theta(2^u)$ recursive steps.  We therefore separate three ingredients:
\[
 \boxed{\text{state width}\ \times\ \text{descent}\ \times\ \text{branch semantics}.}
\]
Width bounds mutable state and clock size.  Descent converts clock size into path length.  Deterministic, existential, universal, or alternating aggregation supplies the Boolean computation-tree semantics; counting is treated as a further aggregation extension.  The local transition algebra does not change when these resource coordinates change.

The formal object is the bounded-state step-recursion dynamics of Definition~\ref{def:bsr}.  Each instance fixes an effective descent $\rho$, a width $u$, and finitely many local maps and guards.  Lemma~\ref{lem:realize} shows that every deterministic instance is realized by one ordinary bounded step recursion over a fixed finite numerical basis.  Independently, Boolean aggregation gives machine semantics to the same local data, and every fixed finite-branching machine transition compiles back into those maps.  Branching is thus an orthogonal aggregation rule; it is not identified with a single-value recursion call.

For generalized-inverse descents, the threshold law from \cite{Osipov2026} converts generator growth into recursion depth.  The present paper uses that construction as input and develops its machine-resource consequences.  Theorem~\ref{thm:orbit} characterizes exactly which increasing sequences occur as generalized-inverse generator orbits.  Theorem~\ref{thm:exactprofile} is the central machine theorem: after the controlled width closure needed to absorb the fixed bit-cost of local operations, the $Q$-semantics of step recursion is exactly the machine class on the depth--width profiles
\[
 (\delta_\rho(u),u).
\]
Equivalently, the construction factors resource behavior through
\[
 \text{generator orbit}\longrightarrow\text{descent depth}\longrightarrow
 (\delta_\rho(u),u)\longrightarrow\text{machine resource family}.
\]

This factorization also gives a comparison theory for descents.  Profile domination induces a preorder allowing width change; its quotient removes distinctions absorbable by enlarging mutable state.  Some pointwise depth curves therefore collapse, whereas Theorem~\ref{thm:quotient-hierarchy} gives an infinite strict chain on polynomial widths between the canonical polynomial- and exponential-depth profiles.  Because strictness is tested against every polynomial width reparameterization, these are genuinely different descent resources rather than artifacts of a width encoding.

General simultaneous time--space algebras already exist in the Wagner tradition \cite{Wagner1979,Clote1999}, while other recursion schemes treat nondeterminism, alternation, and counting \cite{Oitavem2013,Oitavem2022,Skapinakis2024,VollmerWagner1996,DalLagoKahleOitavem2022}, and discrete-ODE approaches obtain complexity characterizations by syntactic restrictions \cite{BournezDurand2019,AntonelliSkapinakis2026}.  Our novelty is therefore not another general time--space characterization.  The local transition schema is fixed while width and descent vary independently, making descents comparable modulo width reparameterization and producing a quotient with both collapses and strict intermediate levels.  Width remains explicit, so this is a resource representation rather than an implicit complexity system; unrestricted numerical composition is a separate issue \cite{AraiEguchi2008}.

Classical complexity classes are checks rather than definitions: the canonical rows recover polynomial time and space, logarithmic-space analogues, fixed exponential towers, alternation, and $\#\Pclass$ from the general profile theorems.

\section{Definitions}

\subsection{Machines and complexity notation}

Inputs are strings $x\in\{0,1\}^*$ and $N:=|x|+2$.  Machines have read-only indexed input and finitely many work tapes.  A configuration stores the finite control, work tapes and heads, and $O(\log N)$-bit input pointers.  We use the usual decider convention throughout: every computation branch halts.

There are four branch modes $Q\in\{\D,\Ex,\All,\Alt\}$.  Deterministic machines have one successor. Existential machines accept if some branch accepts. Universal machines accept if every branch accepts. Alternating machines label each nonterminal configuration by $\exists$ or $\forall$ and use the corresponding rule.

We write $\DTIME(t),\NTIME(t),\ATIME(t)$ and $\DSPACE(s),\NSPACE(s),\ASPACE(s)$ for the resulting deterministic, existential, and alternating decider classes. Universal classes are the complements of the corresponding existential classes.  For a class $\mathcal C$, $\co\mathcal C:=\{L:\overline L\in\mathcal C\}$.  For a branch mode $Q$, write $\TISP_Q(t,s)$ for languages decided in mode $Q$ with path time $O(t)$ and work space $O(s)$ under this convention.
All resource inequalities and domination statements are understood asymptotically.  Finite exceptional input lengths are harmless: a machine can hard-code them in its finite control, and a step-recursion term can handle them by a fixed finite case distinction in the initialization.  Thus a bound that holds for all sufficiently large $N$ suffices throughout.

\begin{lemma}[Total-space path bound]\label{lem:total}
Let $M$ be a total finite-branching machine of any branch mode $Q$ using at most $s(N)$ work cells, where $s(N)\ge \lceil\log_2 N\rceil+1$.  Every branch has length at most $2^{O(s(N))}$.
\end{lemma}

\begin{proof}
A complete configuration has $O(s(N)+\log N)=O(s(N))$ bits, hence there are at most $2^{O(s(N))}$ configurations.  A branch cannot repeat one: replaying the same finite sequence of successor choices from the repeated configuration would produce an infinite branch, contrary to totality.  Thus every branch is simple in the finite configuration graph and has the stated length.
\end{proof}

Put
\[
 \exp_0(n)=n,
 \qquad
 \exp_{k+1}(n)=2^{\exp_k(n)}.
\]
For $k\ge0$ define
\[
 k\text{-}\mathsf{EXPTIME}:=\bigcup_{d\ge1}\DTIME(\exp_k(N^d)),
 \qquad
 k\text{-}\mathsf{EXPSPACE}:=\bigcup_{d\ge1}\DSPACE(\exp_k(N^d)),
\]
and
\[
 \mathsf N k\text{-}\mathsf{EXPTIME}:=\bigcup_{d\ge1}\NTIME(\exp_k(N^d)),
 \qquad
 \mathsf{coN}k\text{-}\mathsf{EXPTIME}:=\co(\mathsf N k\text{-}\mathsf{EXPTIME}).
\]
We abbreviate
\[
 \Pclass:=0\text{-}\mathsf{EXPTIME},\quad
 \NP:=\mathsf N0\text{-}\mathsf{EXPTIME},\quad
 \PSPACE:=0\text{-}\mathsf{EXPSPACE},
\]
\[
 \EXP:=1\text{-}\mathsf{EXPTIME},\quad
 \NEXP:=\mathsf N1\text{-}\mathsf{EXPTIME},\quad
 \EXPSPACE:=1\text{-}\mathsf{EXPSPACE}.
\]
Also
\[
 \Lclass=\DSPACE(O(\log N)),\qquad
 \NL=\NSPACE(O(\log N)),
\]
\[
 \Eclass=\DTIME(2^{O(N)}),\qquad
 \NE=\NTIME(2^{O(N)}),\qquad
 \ESPACE=\DSPACE(2^{O(N)}).
\]

\subsection{Descent and depth}

An \emph{effective descent} is a nondecreasing map $\rho:\N\to\N$ such that $\rho(0)=0$, $\rho(y)<y$ for $y>0$, and, on a $b$-bit input, $\rho(y)$ is computable in $b^{O(1)}$ time and $O(b)$ space.  Its depth is
\[
 D_\rho(y):=\min\{t:\rho^{[t]}(y)=0\}.
\]
All descents in this paper are effective.  The step descents motivating the construction arise as generalized inverses.  Following~\cite{Osipov2026}, let $\varphi:\N\to\N$ be strictly increasing with $\varphi(z)\ge z+1$ for every $z$, and define
\[
 \rho_\varphi(y):=\min\{z\in\N:\varphi(z)\ge y\}.
\]
The minimum exists and $\rho_\varphi(y)<y$ for $y>0$.  Whenever $\rho_\varphi$ is effective, it is therefore an effective descent in the preceding sense.  We restate the threshold law from~\cite{Osipov2026} for self-containment.

\begin{theorem}[Generalized-inverse threshold law]\label{thm:inverse-depth}
For every $y\in\N$,
\[
 D_{\rho_\varphi}(y)
 =\min\{t\ge0:y\le \varphi^{[t]}(0)\}.
\]
Consequently, for every width term $u$,
\[
 \delta_{\rho_\varphi}(u;N)
 =\min\{t\ge0:2^{u(N)}-1\le \varphi^{[t]}(0)\}.
\]
Thus the iteration growth of $\varphi$ determines the intrinsic path budget of the associated step descent.
\end{theorem}

\begin{proof}
For every $a,y\in\N$, monotonicity of $\varphi$ gives the adjunction
\[
 \rho_\varphi(y)\le a
 \quad\Longleftrightarrow\quad
 y\le\varphi(a).
\]
Indeed, the left side says that some $z\le a$ satisfies $\varphi(z)\ge y$, which by monotonicity is equivalent to $\varphi(a)\ge y$.  Iterating this equivalence yields
\[
 \rho_\varphi^{[t]}(y)=0
 \quad\Longleftrightarrow\quad
 y\le\varphi^{[t]}(0).
\]
Taking the least such $t$ proves the first identity. Substituting $K_u(N)=2^{u(N)}-1$ gives the second.
\end{proof}

\begin{theorem}[Orbit characterization]\label{thm:orbit}
Let
\[
 0=a_0<a_1<a_2<\cdots
\]
be an increasing integer sequence and put $\Delta_t:=a_{t+1}-a_t$.  The following are equivalent:
\begin{enumerate}[label=(\roman*),leftmargin=2em]
 \item there is a strictly increasing map $\varphi:\N\to\N$ with $\varphi(z)\ge z+1$ and
 \[
  \varphi^{[t]}(0)=a_t\qquad(t\ge0).
 \]
 \item the gaps are nondecreasing: $\Delta_{t+1}\ge\Delta_t$ for every $t$.
\end{enumerate}
For every such sequence, the explicit generator constructed below has generalized-inverse depth
\[
 D_{\rho_{\varphi_a}}(y)=\min\{t:y\le a_t\}.
\]
If this generalized inverse is effective, the sequence therefore defines an admissible step-recursion resource profile.
\end{theorem}

\begin{proof}
Suppose first that $a_t=\varphi^{[t]}(0)$ for a strictly increasing integer map $\varphi$.  Strict increase on $\N$ implies
\[
 \varphi(b)-\varphi(a)\ge b-a\qquad(b>a),
\]
because each unit increase of the input raises the integer output by at least one.  Hence
\[
 \Delta_{t+1}
 =\varphi(a_{t+1})-\varphi(a_t)
 \ge a_{t+1}-a_t
 =\Delta_t.
\]
Thus (i) implies (ii).

Conversely, assume that the gaps are nondecreasing.  For the unique $t$ with $a_t\le z<a_{t+1}$, put
\[
 \varphi_a(z):=a_{t+1}+(z-a_t).
\]
Within each interval this is strictly increasing and exceeds $z$ by $\Delta_t\ge1$.  At the boundary, the last value before $a_{t+1}$ is
\[
 a_{t+1}+\Delta_t-1
 \le a_{t+1}+\Delta_{t+1}-1
 =a_{t+2}-1,
\]
so the next value $\varphi_a(a_{t+1})=a_{t+2}$ is larger.  Thus $\varphi_a$ is strictly increasing, satisfies $\varphi_a(z)\ge z+1$, and $\varphi_a(a_t)=a_{t+1}$ gives the orbit identity by induction.  The depth formula is Theorem~\ref{thm:inverse-depth}.
\end{proof}

\begin{corollary}[Profiles beyond the canonical rows]\label{cor:intermediate}
For every fixed integer $k\ge1$ there are effective generalized-inverse descents $\rho_k^{\rm sh}$ and $\rho_k^{\rm mid}$ with width-$u$ depths
\[
 \delta_{\rho_k^{\rm sh}}(u)=\Theta\!\left(u^{1/k}\right),
 \qquad
 \delta_{\rho_k^{\rm mid}}(u)=\Theta\!\left(2^{u/k}\right).
\]
For $k>1$, pointwise in width, the first family is strictly shallower than the divisive $\Theta(u)$ row, while the second lies strictly between that row and the predecessor $\Theta(2^u)$ row.
\end{corollary}

\begin{proof}
Apply Theorem~\ref{thm:orbit} to $a_t=2^{t^k}-1$ and to $a_t=t^k$.  Both have nondecreasing gaps.  Theorem~\ref{thm:inverse-depth} at $K_u=2^u-1$ gives respectively the least $t$ with $u\le t^k$ and the least $t$ with $2^u-1\le t^k$, yielding the displayed orders.  For these two sequences the interval index is computable from a fixed-degree integer root (and, in the exponential case, binary length).  Moreover, if $a_{t+1}\le y\le a_{t+2}$, then
\[
 \rho_{\varphi_a}(y)=
 \begin{cases}
  a_t+(y-a_{t+1}),&y\le a_{t+1}+\Delta_t-1,\\
  a_{t+1},&y>a_{t+1}+\Delta_t-1,
 \end{cases}
\]
with $\rho_{\varphi_a}(0)=0$.  Thus the required arithmetic, comparisons, and fixed-degree roots run in polynomial time and $O(b)$ space on $b$-bit arguments. The descents are effective.
\end{proof}

The canonical descents are generalized inverses of
\[
 \varphi_{0,l}(z)=z+l,
 \qquad
 \varphi_{1,l}(z)=2^l z+(2^l-1).
\]
Indeed, their generalized inverses are respectively truncated subtraction by $l$ and division by $2^l$.  Fix a stride $l\ge1$ and define
\[
 \rho_{0,l}(y)=\max\{y-l,0\},
 \qquad
 \rho_{1,l}(y)=\left\lfloor\frac{y}{2^l}\right\rfloor.
\]

\begin{lemma}[Exact depth]\label{lem:depth}
For $y\in\N$,
\[
 D_{\rho_{0,l}}(y)=\left\lceil\frac yl\right\rceil,
 \qquad
 D_{\rho_{1,l}}(y)=\left\lceil\frac{\lambda(y)}l\right\rceil,
\]
where $\lambda(y):=\lceil\log_2(y+1)\rceil$.
\end{lemma}

\begin{proof}
The generator iterates are
\[
 \varphi_{0,l}^{[t]}(0)=lt,
 \qquad
 \varphi_{1,l}^{[t]}(0)=2^{lt}-1.
\]
Apply Theorem~\ref{thm:inverse-depth}.  The conditions $y\le lt$ and $y\le2^{lt}-1$ give the two displayed formulas.
\end{proof}

\subsection{Finite width syntax}

Width terms are generated by the finite grammar
\[
 u::=\mathbf L\mid\mathbf N\mid(u+u)\mid(u\cdot u)\mid2^u,
\]
with $\mathbf L(N)=\lceil\log_2N\rceil+1$ and $\mathbf N(N)=N$.  A term is a capacity annotation, not an oracle for an arbitrary resource function.  Its clock is
\[
 K_u(N):=2^{u(N)}-1.
\]
For an effective descent $\rho$ define its depth profile
\[
 \delta_\rho(u;N):=D_\rho(K_u(N)).
\]

\begin{lemma}[Depth--width identity]\label{lem:duality}
For every width term $u$,
\[
 \delta_{\rho_{1,l}}(u;N)=\left\lceil\frac{u(N)}l\right\rceil,
 \qquad
 \delta_{\rho_{0,l}}(u;N)=\left\lceil\frac{2^{u(N)}-1}{l}\right\rceil.
\]
Hence the two profiles are respectively $\Theta(u(N))$ and $\Theta(2^{u(N)})$.
\end{lemma}

\begin{proof}
The clock $K_u$ has binary length exactly $u(N)$. Apply Lemma~\ref{lem:depth}.
\end{proof}

We use the following width families.  $\mathcal W_{\log}$ consists of terms generated from $\mathbf L$ by addition.  $\mathcal W_0$ consists of terms generated from $\mathbf N$ by addition and multiplication.  For $k\ge1$, $\mathcal W_k$ consists of terms generated from $\mathbf N$ by addition, multiplication, and at most $k$ nested occurrences of $2^{(\cdot)}$ on every syntax branch.  Finally $\mathcal W_E$ consists of sums and products of terms $2^v$ where $v$ is obtained from $\mathbf N$ by addition only.

Every width term satisfies $u(N)\ge\mathbf L(N)$.  Call a width family $\mathcal W$ \emph{scaling-closed} if for every $u\in\mathcal W$ and integer $c\ge1$ some $v\in\mathcal W$ satisfies $v(N)\ge c\,u(N)$ for all $N$. Call it \emph{power-cofinal} if for every $u\in\mathcal W$ and fixed integer $d\ge1$ some $v\in\mathcal W$ satisfies $u(N)^d=O(v(N))$.  Hence scaling closure already absorbs the $O(\log N)$ input-pointer overhead.

\begin{lemma}[Width scales]\label{lem:widths}
The terms of $\mathcal W_{\log}$ are $O(\log N)$ and dominate every fixed multiple of $\log N$. The family $\mathcal W_0$ is polynomially bounded and contains $N^d$ for every fixed $d$. Every $\mathcal W_E$ term is $2^{O(N)}$ and the family is cofinal for that scale. Finally, $\mathcal W_k$ is cofinal for $\exp_k(N^{O(1)})$.  Every displayed family is scaling-closed, while $\mathcal W_0,\mathcal W_E,$ and $\mathcal W_k$ are power-cofinal.  Every fixed $u$ and its clock $K_u$ can be initialized in $u(N)^{O(1)}$ time and $O(u(N))$ space.
\end{lemma}

\begin{proof}
Induct on the fixed syntax tree.  Addition and multiplication preserve polynomial scale. At a fixed positive tower height, fixed products and powers are absorbed by changing a constant or the bottom polynomial.  Conversely the grammar contains terms dominating $N^d$, $2^{cN}$, and $\exp_k(N^d)$ for all fixed $c,d$.  The stability claims follow from the same observations. The canonical representatives $N^d$, $2^{cN}$, and $\exp_k(N^d)$ used below are constructible at their displayed scales. To initialize $K_u$, evaluate the fixed syntax tree with binary counters and then write $u(N)$ one-bits.
\end{proof}

\subsection{Bounded-state step-recursion dynamics}

A \emph{state of width $u$} is a fixed tuple of binary stacks, finite control fields, and input pointers.  Each stack contains at most $u(N)$ data bits. Each pointer ranges over $\{0,\ldots,|x|\}$ and therefore uses at most $\mathbf L(N)$ bits.  The input $x$ is read-only and is not part of the mutable state.  Since the tuple arity is fixed and $u\ge\mathbf L$, every state has $O(u(N))$ mutable bits.  The finite local basis contains Boolean operations and conditionals, push-$0$, push-$1$, empty/top/pop on bounded stacks, pointer increment/decrement and zero test, the input length $|x|$ available in binary, and $\operatorname{Read}(x,i)$.  A push is allowed exactly when the current stack has fewer than $u(N)$ data bits. Otherwise the update enters a distinguished rejecting overflow state.  Empty-stack pop preserves the empty stack and its top symbol is the blank bit $0$.  Pointer decrement at $0$ stays at $0$, while an increment beyond $|x|$ enters the same overflow state.  Thus every primitive is total and width-preserving.  A \emph{local map} is a fixed composition of these primitives, with read-only access to $x$, and is computable in $u(N)^{O(1)}$ time and $O(u(N))$ space.

For a finite nonempty tuple of bits $\bar b$, define
\[
 \Gamma_{\D}(b_0)=b_0,
 \qquad
 \Gamma_{\Ex}(\bar b)=\bigvee_i b_i,
 \qquad
 \Gamma_{\All}(\bar b)=\bigwedge_i b_i,
\]
and let $\Gamma_{\Alt,\theta}$ be $\vee$ or $\wedge$ according as $\theta\in\{\exists,\forall\}$.

\begin{definition}[Bounded-state step-recursion dynamics and Boolean semantics]\label{def:bsr}
Fix an effective descent $\rho$, width $u$, and local maps in the preceding sense: initialization $I(x)$, successors $T_0(x,c),\ldots,T_{q-1}(x,c)$, guards $\eta_i(x,c)$, halting and accepting tests $h(x,c),a(x,c)$, and a type label $\theta(x,c)\in\{\exists,\forall\}$.  Write $\mathfrak D=(\rho,u,I,\bar T,\bar\eta,h,a,\theta)$ for this \emph{bounded-state step-recursion dynamics}.  At least one guard is true at every nonhalting state. The dynamics $\mathfrak D$ is \emph{deterministic} when exactly one is true.

For $Q\in\{\D,\Ex,\All,\Alt\}$, with $Q=\D$ only for deterministic $\mathfrak D$, define
\[
 B^Q_{\mathfrak D}(x,c,y)=
 \begin{cases}
  a(x,c),&h(x,c)=1,\\
  0,&h(x,c)=0\text{ and }y=0,\\
  \Gamma_Q\bigl(B^Q_{\mathfrak D}(x,T_i(x,c),\rho(y)):\eta_i(x,c)=1\bigr),&\text{otherwise},
 \end{cases}
\]
where $\Gamma_Q=\Gamma_{\Alt,\theta(x,c)}$ in alternating mode, and put $\BSR^Q(\mathfrak D)(x):=B^Q_{\mathfrak D}(x,I(x),K_u(N))$.  Thus $\mathfrak D$ is fixed before its Boolean semantics $Q$ is chosen.
\end{definition}

Lemma~\ref{lem:realize} realizes the deterministic semantics by one numerical bounded step recursion.  The other Boolean modes aggregate the same dynamics. Counting is introduced later as a separate extension and is not part of the Boolean equivalence theorem.

To connect the bounded-state notation to numerical step recursion, encode $x$ by
\[
 \nu(x):=2^{|x|}+\operatorname{val}(x),\qquad \lambda(\nu(x))=|x|+1.
\]
Fix the following finite non-projection numerical basis. In addition, all standard argument projections are available:
\[
 \mathcal C_{\rm bit}=\{0,1,+,\times,E,\operatorname{pred},\lambda,
 \operatorname{Cond},\operatorname{leq},\operatorname{push}_0,\operatorname{push}_1,
 \operatorname{pop},\operatorname{last},\pi,\pi_0,\pi_1,\operatorname{Read}\}.
\]
Here $E(z)=2^z$, $\operatorname{pred}(z)=\max\{z-1,0\}$,
$\operatorname{Cond}(b,r,s)=r$ if $b=1$ and $s$ otherwise, and
$\operatorname{leq}(r,s)$ is $1$ exactly when $r\le s$.  Also
$\operatorname{push}_b(z)=2z+b$, $\operatorname{pop}(z)=\lfloor z/2\rfloor$, and
$\operatorname{last}(z)=z\bmod2$.  The fixed pairing $\pi$ has projections
$\pi_0,\pi_1$ and satisfies
$\lambda(\pi(r,s))=O(\lambda(r)+\lambda(s)+1)$.
For $z>0$, write its canonical binary representation as $1x$ and define
$\operatorname{Read}(z,i)$ to be the $i$th bit of the payload $x$ under a fixed zero-based convention, or $0$ when $i$ is out of range. Put $\operatorname{Read}(0,i)=0$.
Thus $\operatorname{Read}(\nu(x),i)$ is exactly read-only access to the input string $x$.  Put
\[
 N_z:=\lambda(z)+1,\qquad L_z:=\lambda(\operatorname{pred}(N_z))+1,
\]
and define $U_u(z)$ from the syntax of $u$ using the leaves $L_z,N_z$ and the constructors $+,\times,E$.  Then $U_u(\nu(x))=u(|x|+2)$.

Encode a stack $b_1\cdots b_m$ by the integer with binary representation $1b_1\cdots b_m$. Hence the empty stack is $1$ and its data length is $\lambda(s)-1$.  Define the total semantic stack operations
\[
 \operatorname{Top}(s)=
 \begin{cases}0,&s\le1,\\ \operatorname{last}(s),&s>1,\end{cases}
 \qquad
 \operatorname{Pop}(s)=
 \begin{cases}1,&s\le1,\\ \operatorname{pop}(s),&s>1,\end{cases}
\]
and let $\operatorname{Empty}(s)=1$ exactly when $s=1$.  A push is legal exactly when
$\lambda(s)\le U_u(z)$, equivalently when the stack currently has fewer than $u(N)$ data bits on a valid input code.  Fixed tuples are encoded by iterated $\pi$.  Hence every local map and width test used above has a numerical representative in the composition closure of $\mathcal C_{\rm bit}$ and the standard argument projections.  This basis is used only for literal algebraic realization, not for a complexity upper bound.

For an effective descent $\rho$, define $\mathcal A_\rho(\mathcal C_{\rm bit})$ by finite stages.  Let
\[
 A_0=\operatorname{Cl}_{\circ}(\mathcal C_{\rm bit}),
\]
where $\operatorname{Cl}_{\circ}$ denotes closure under composition together with the standard argument-projection schema.  By the preceding explicit constructions, every $U_u$, $\operatorname{Top}$, $\operatorname{Pop}$, and $\operatorname{Empty}$ belongs to $A_0$.
Given $A_r$, let $A_{r+1}$ be the composition closure of $A_r$ together with every function $f$ obtained from $g,h,b\in A_r$ by
\[
 f(\bar z,0)=g(\bar z),\qquad
 f(\bar z,y)=h\bigl(\bar z,\rho(y),f(\bar z,\rho(y))\bigr)\quad(y>0),
\]
subject to the pointwise bound $f(\bar z,y)\le b(\bar z,y)$.  Put
\[
 \mathcal A_\rho(\mathcal C_{\rm bit})=\bigcup_{r<\omega}A_r.
\]
This is ordinary bounded step recursion with the required ``earlier data'' condition made explicit.

\begin{lemma}[Numerical realization]\label{lem:realize}
Let $\mathfrak D$ be a deterministic width-$u$ dynamics from Definition~\ref{def:bsr}.  There is a Boolean function $F_{\mathfrak D}\in\mathcal A_\rho(\mathcal C_{\rm bit})$ such that
\[
 F_{\mathfrak D}(\nu(x))=1\quad\Longleftrightarrow\quad \BSR^\D(\mathfrak D)(x)=1.
\]
The construction uses exactly one bounded step recursion. Its initialization, transition, bound, and final tests lie in $A_0$.
\end{lemma}

\begin{proof}
Encode the fixed tuple of stacks, control fields, and pointers by iterated $\pi$.  For some constant $C_{\mathfrak D}$, every legal state code has bit length at most
\[
 M_{\mathfrak D}(z):=C_{\mathfrak D}\bigl(U_u(z)+L_z+1\bigr).
\]
By the preceding encoding, the total initialization $G^0_{\mathfrak D}(z)$, deterministic update $S^0_{\mathfrak D}(z,r)$, halting test $H_{\mathfrak D}$, acceptance test $A_{\mathfrak D}$, and validity checks all lie in $A_0$. Malformed or overflowing states are sent to a fixed rejecting state, and halting states are made absorbing.  Let
\[
 Q_{\mathfrak D}(z,y):=\operatorname{pred}(E(M_{\mathfrak D}(z))).
\]
Clipping $G^0_{\mathfrak D}$ and $S^0_{\mathfrak D}$ to this cap gives $G_{\mathfrak D},S_{\mathfrak D}\in A_0$ that agree with every legal computation and are everywhere bounded by $Q_{\mathfrak D}$.  Hence
\[
 R_{\mathfrak D}(z,0)=G_{\mathfrak D}(z),\qquad
 R_{\mathfrak D}(z,y)=S_{\mathfrak D}\bigl(z,R_{\mathfrak D}(z,\rho(y))\bigr)\quad(y>0)
\]
is one admissible bounded step recursion, so $R_{\mathfrak D}\in A_1$.  Induction on $D_\rho(y)$ shows that it is the encoded state after that many transition opportunities.  Since the clock $\operatorname{pred}(E(U_u(z)))$ belongs to $A_0$,
\[
 F_{\mathfrak D}(z)=H_{\mathfrak D}(z,R_{\mathfrak D}(z,K))A_{\mathfrak D}(z,R_{\mathfrak D}(z,K)),
 \qquad K:=\operatorname{pred}(E(U_u(z))),
\]
belongs to $\mathcal A_\rho(\mathcal C_{\rm bit})$ and equals the deterministic BSR value on every valid input code.
\end{proof}

For a width family $\mathcal W$, $\SR^Q_{\rho}[\mathcal W]$ is the class of languages accepted by the $Q$-semantics of dynamics from Definition~\ref{def:bsr} with $u\in\mathcal W$.  We abbreviate $\SR^Q_{j,l}[\mathcal W]:=\SR^Q_{\rho_{j,l}}[\mathcal W]$.  Completed branching predicates are not allowed as new local guards. Otherwise the definition would introduce oracle access.

\begin{lemma}[Configuration coding]\label{lem:config}
For every fixed deterministic, existential, universal, or alternating multitape Turing machine $M$, its one-step successor relation is represented by fixed local maps and guards.  If a configuration of $M$ uses $s(N)$ work cells, it fits in width $u$ whenever
\[
 u(N)\ge c_M(s(N)+\mathbf L(N))
\]
for a constant $c_M$ depending only on $M$.
\end{lemma}

\begin{proof}
Represent each work tape by the two stacks on either side of its head, the scanned symbol, and finite control.  For an input head at position $i$, store the pair $(i,|x|-i)$. Zero tests on the two components detect the two input boundaries, and a head move increments one component while decrementing the other.  These are $O(\log N)$ pointers, and $\operatorname{Read}(x,i)$ supplies the scanned input bit.  One Turing transition therefore changes only finite control, one scanned symbol, finitely many stack ends, and input pointers, hence is a fixed local map.
\end{proof}

\begin{remark}[Machine-model robustness]\label{rem:robust}
The results extend to fixed finite-branching models with $O(s+\mathbf L(N))$-bit configurations and fixed local one-step maps, including standard two-stack and bounded-register/RAM presentations \cite{vanEmdeBoas1990}; no fine-grained unit-cost claim is made for counter/Minsky machines.  Reversible local successors give a reversible submodel, with deterministic space preserved asymptotically \cite{LangeMcKenzieTapp2000}.
\end{remark}

\section{Resource factorization: algebra and machines}

For a width family $\mathcal W$, write
\[
 \DTIME(\mathcal W):=\bigcup_{u\in\mathcal W}\DTIME(u(N)),
 \qquad
 \DSPACE(\mathcal W):=\bigcup_{u\in\mathcal W}\DSPACE(u(N)),
\]
and analogously for the other modes.  To shorten statements, set
\[
 \mathsf{TIME}_{\D}=\DTIME,\quad
 \mathsf{TIME}_{\Ex}=\NTIME,\quad
 \mathsf{TIME}_{\All}=\co\NTIME,\quad
 \mathsf{TIME}_{\Alt}=\ATIME,
\]
\[
\begin{aligned}
 \mathsf{SPACE}_{\D}&=\DSPACE, &
 \mathsf{SPACE}_{\Ex}&=\NSPACE,\\
 \mathsf{SPACE}_{\All}&=\co\NSPACE, &
 \mathsf{SPACE}_{\Alt}&=\ASPACE.
\end{aligned}
\]

The central object can therefore be summarized as
\[
\begin{array}{c}
 \boxed{\text{bounded-state step-recursion dynamics}}\\[2pt]
 \swarrow\ \text{algebraic realization}
 \qquad
 \text{machine semantics}\ \searrow\\[2pt]
 \text{ordinary bounded step recursion}
 \qquad
 \text{finite-branching Boolean machine profiles}.
\end{array}
\]
The left arrow is the literal deterministic algebraic realization of Lemma~\ref{lem:realize}. The right equivalence is established below after width-family closure.  The left arrow is intentionally not an equality: Lemma~\ref{lem:realize} embeds the bounded-state deterministic core into $\mathcal A_\rho(\mathcal C_{\rm bit})$ and does not identify the full algebra $\mathcal A_\rho(\mathcal C_{\rm bit})$ with a machine class.

For an effective descent $\rho$ and width family $\mathcal W$, define the machine class on its intrinsic profile by
\[
 \mathsf{Prof}_Q(\rho,\mathcal W)
 :=\bigcup_{u\in\mathcal W}
 \TISP_Q\!\left(\delta_\rho(u;N),u(N)\right).
\]
Call $\mathcal W$ \emph{$\rho$-profile-closed} if for every $u\in\mathcal W$, integers $d\ge0$ and $C\ge1$, there is $v\in\mathcal W$ such that, for all sufficiently large $N$,
\[
 v(N)\ge C u(N),
 \qquad
 \delta_\rho(v;N)\ge C\,\delta_\rho(u;N)u(N)^d.
\]
This sufficient condition is only an overhead-absorption hypothesis: it absorbs the bit-cost of a fixed local algebra after width reparameterization and is neither a new complexity scale nor the source of the hierarchy.  The structural point is that the same local schema factors for arbitrary effective descent, exposing descents to a common comparison relation. Lemma~\ref{lem:profileclosure} gives natural sufficient conditions for the canonical rows.

The elementary fact that a bounded configuration graph can simulate a machine is standard.  The role of the next theorem is instead to isolate which coordinate contributes path length and which contributes storage, so that changing the descent while keeping the local transition basis fixed becomes a well-defined resource operation.

\begin{theorem}[Step-recursion time--space profile]\label{thm:profile}
Fix an effective descent $\rho$ and a branch mode $Q$.
\begin{enumerate}[label=(\roman*),leftmargin=2em]
 \item The $Q$-semantics of a width-$u$ dynamics from Definition~\ref{def:bsr} belongs to
 \[
  \TISP_Q\!\left(\delta_\rho(u;N)\,u(N)^{O(1)},\;O(u(N))\right).
 \]
 The pair $(\delta_\rho(u),u)$ is the intrinsic depth--width profile of the dynamics. The displayed $u^{O(1)}$ factor is the machine cost of the fixed local algebra.
 \item Let $M$ be a total machine of branch mode $Q$, using at most $s(N)$ work cells and with branch length at most $t(N)$.  If a width term $u$ satisfies
 \[
  u(N)\ge c_M(s(N)+\mathbf L(N)),
  \qquad
  \delta_\rho(u;N)\ge t(N),
 \]
 then $M$ is represented by the $Q$-semantics of a width-$u$ dynamics from Definition~\ref{def:bsr}.
\end{enumerate}
\end{theorem}

\begin{proof}
For (i), recursion decreases the clock at every recursive transition, so every branch has at most $\delta_\rho(u;N)$ transitions and $\delta_\rho(u;N)+1$ nodes.  Since width terms have positive clocks, $\delta_\rho(u;N)\ge1$.  At each node, the fixed tuple of $q$ guards is evaluated and the machine branches exactly to the enabled successors. In deterministic mode there is exactly one.  Since $q$ is fixed, guard evaluation, computing $\rho(y)$ on the $O(u)$-bit clock, and applying the selected fixed local map together cost $u^{O(1)}$ time and $O(u)$ reusable workspace, by effectiveness of $\rho$ and $u\ge\mathbf L$.  Thus existential, universal, and alternating aggregation is implemented by the corresponding ordinary machine branching, with no oracle predicate hidden in a guard.

For (ii), use Lemma~\ref{lem:config} for the initialization, successors, guards, halting test, acceptance test, and alternating type.  The width hypothesis stores every configuration, and the depth hypothesis prevents clock exhaustion before any machine branch halts.  The recursion therefore has exactly the same computation tree and branch aggregation as $M$.
\end{proof}

\begin{theorem}[Resource factorization theorem]\label{thm:exactprofile}
If $\mathcal W$ is $\rho$-profile-closed, then for every Boolean branch mode $Q\in\{\D,\Ex,\All,\Alt\}$,
\[
 \boxed{\SR^Q_{\rho}[\mathcal W]
 =\mathsf{Prof}_Q(\rho,\mathcal W).}
\]
This is equality after union over the width family, not a pointwise same-width time identity.  The local factor $u^{O(1)}$ from Theorem~\ref{thm:profile}(i) is absorbed by passing to a larger $v\in\mathcal W$.
\end{theorem}

\begin{proof}
For soundness, Theorem~\ref{thm:profile}(i) gives path time $O(\delta_\rho(u)u^d)$ and space $O(Cu)$ for fixed constants $d,C$ depending only on the dynamics.  Profile closure supplies $v\in\mathcal W$ with $v\ge Cu$ and $\delta_\rho(v)\ge C\delta_\rho(u)u^d$, so the language lies in $\mathsf{Prof}_Q(\rho,\mathcal W)$.

Conversely, let a total machine $M$ witness membership in $\TISP_Q(\delta_\rho(u),u)$ for some $u\in\mathcal W$.  Its hidden time and space constants and the configuration constant of Lemma~\ref{lem:config} are fixed.  Profile closure with $d=0$ gives $v\in\mathcal W$ large enough that $v$ stores every configuration and $\delta_\rho(v)$ dominates the branch-time bound.  Theorem~\ref{thm:profile}(ii) then compiles $M$ into the $Q$-semantics of a width-$v$ dynamics.
\end{proof}

\begin{corollary}[Generator-to-resource theorem]\label{cor:generator-profile}
Let $\rho_\varphi$ be an effective generalized-inverse descent as in Theorem~\ref{thm:inverse-depth}, and let $\mathcal W$ be $\rho_\varphi$-profile-closed.  Then, for every Boolean branch mode $Q$,
\[
 \SR^Q_{\rho_\varphi}[\mathcal W]
 =\bigcup_{u\in\mathcal W}
 \TISP_Q\!\left(
   \min\{t:2^{u(N)}-1\le\varphi^{[t]}(0)\},
   u(N)
 \right).
\]
Hence the iteration growth of the step generator $\varphi$, together with width, determines the represented machine resource family after the fixed local overhead is absorbed by passing to larger widths inside $\mathcal W$.
\end{corollary}

\begin{proof}
Combine Theorem~\ref{thm:inverse-depth} with Theorem~\ref{thm:exactprofile}.
\end{proof}

\begin{corollary}[Profile order theorem]\label{cor:compare}
Let $\rho,\sigma$ be effective descents and let $\mathcal W$ be profile-closed for both.  If for every $u\in\mathcal W$ some $v\in\mathcal W$ satisfies
\[
 u(N)=O(v(N)),\qquad \delta_\rho(u;N)=O(\delta_\sigma(v;N)),
\]
then $\SR^Q_\rho[\mathcal W]\subseteq\SR^Q_\sigma[\mathcal W]$ for every Boolean branch mode $Q$.  Mutual domination gives equality.
\end{corollary}

\begin{proof}
By Theorem~\ref{thm:exactprofile}, a $\rho$-language has a machine using $O(u)$ space and $O(\delta_\rho(u))$ path time.  The displayed domination places that machine on a $\sigma$ profile. Apply Theorem~\ref{thm:exactprofile} again.
\end{proof}

\subsection{Profile preorder and quotient}

Fix a width family $\mathcal W$ and let $\mathfrak R_{\mathcal W}$ be the effective descents $\rho$ for which $\mathcal W$ is $\rho$-profile-closed.  On $\mathfrak R_{\mathcal W}$ write $\rho\preceq_{\mathcal W}\sigma$ for the domination relation in Corollary~\ref{cor:compare}.  This is a preorder: reflexivity is immediate and transitivity follows by composing the width witnesses.  Let $\rho\asymp_{\mathcal W}\sigma$ denote its mutual-domination equivalence, and write $\rho\prec_{\mathcal W}\sigma$ when $\rho\preceq_{\mathcal W}\sigma$ but $\sigma\not\preceq_{\mathcal W}\rho$.  The preorder induces a partial order on $\mathfrak R_{\mathcal W}/\asymp_{\mathcal W}$.  For every $Q$, the map $\rho\mapsto\SR^Q_\rho[\mathcal W]$ is monotone under $\preceq_{\mathcal W}$ and factors through $\asymp_{\mathcal W}$.

The quotient is substantially richer than the two canonical rows.  We first construct an infinite hierarchy beginning at polynomial depth and continuing through quasipolynomial scales.  For a fixed integer $r\ge1$, put
\[
 q_r(t):=\left\lceil \lambda(t+1)^{1/r}\right\rceil,
 \qquad
 \Delta_t^{(r)}:=2^{2^{q_r(t)}},
\]
and define
\[
 a_0^{(r)}:=0,
 \qquad
 a_{t+1}^{(r)}:=a_t^{(r)}+\Delta_t^{(r)}.
\]
The gaps $\Delta_t^{(r)}$ are nondecreasing, so Theorem~\ref{thm:orbit} supplies a generator $\varphi_r$ with orbit $(a_t^{(r)})_{t\ge0}$.  Let
\[
 \tau_r:=\rho_{\varphi_r}.
\]

\begin{lemma}[Saturated orbit evaluation]\label{lem:saturated-orbit}
Fix $r\ge1$.  Given a $b$-bit integer $Y>0$ and $0\le t\le Y+2$, the truncated value
\[
 \min\{a_t^{(r)},Y+1\}
\]
can be computed in time polynomial in $b$ and space $O(b)$.
\end{lemma}

\begin{proof}
For fixed $r$, the value $q_r(i)=q$ precisely on blocks for which
\[
 (q-1)^r<\lambda(i+1)\le q^r.
\]
Thus the partial sum defining $a_t^{(r)}$ can be grouped into at most $q_r(t)=O(b)$ block contributions of the form
\[
 c_q(t)\,2^{2^q},
\]
where the block multiplicity $c_q(t)\le t$ and its endpoints are computable from fixed-degree integer powers and binary shifts, truncated at $t$.  All endpoints and multiplicities therefore use $O(b)$ bits.

It is never necessary to construct a doubly exponential summand whose binary length exceeds the cap.  If $2^q\ge b$, then $2^{2^q}>Y$, so any positive contribution from that block makes the truncated sum equal to $Y+1$.  Otherwise $2^q<b$, and the summand itself has at most $b$ bits.  Capped addition and multiplication at $Y+1$ then evaluate every remaining block using $O(b)$ space and polynomial time.  Since there are only $O(b)$ relevant blocks, the stated bounds follow.
\end{proof}

\begin{theorem}[Infinite hierarchy in the polynomial-width quotient]\label{thm:quotient-hierarchy}
For every fixed $r\ge1$, the descent $\tau_r$ is effective, $\mathcal W_0$ is $\tau_r$-profile-closed, and
\[
 \delta_{\tau_r}(u;N)
 =2^{\Theta((\log_2 u(N))^r)}
 \qquad(u\in\mathcal W_0).
\]
Consequently, for every fixed $l\ge1$,
\[
 [\rho_{1,l}]_{\asymp_{\mathcal W_0}}
 =[\tau_1]_{\asymp_{\mathcal W_0}}
 <[\tau_2]_{\asymp_{\mathcal W_0}}
 <[\tau_3]_{\asymp_{\mathcal W_0}}
 <\cdots
 <[\rho_{0,l}]_{\asymp_{\mathcal W_0}}.
\]
Thus, even after all polynomial width changes are factored out, the descent coordinate has infinitely many inequivalent levels between the divisive and predecessor rows.  The strictness is in the profile quotient; it does not assert language-class separation.
\end{theorem}

\begin{proof}
The gaps are nondecreasing by construction, so Theorem~\ref{thm:orbit} gives $\varphi_r$.  We first verify effectiveness without ever materializing the large orbit values.  Since every gap satisfies $\Delta_i^{(r)}\ge4$, we have $a_t^{(r)}\ge4t$; hence every interval relevant to an input $y$ has index $t\le y$.  The case $y=0$ is immediate, and if $0<y\le a_1^{(r)}$ then $\tau_r(y)=0$.  Otherwise, on a $b$-bit argument $y$, binary search over $0\le t\le y$ uses $O(b)$ trials.  Lemma~\ref{lem:saturated-orbit} evaluates each comparison with $a_{t+1}^{(r)}$ or $a_{t+2}^{(r)}$ after saturation at $y+1$, so the unique interval
\[
 a_{t+1}^{(r)}\le y<a_{t+2}^{(r)}
\]
can be located in polynomial time and $O(b)$ space.  Once it is located, the explicit orbit construction from Theorem~\ref{thm:orbit} gives
\[
 \tau_r(y)=
 \begin{cases}
  a_t^{(r)}+\bigl(y-a_{t+1}^{(r)}\bigr),
    &y\le a_{t+1}^{(r)}+\Delta_t^{(r)}-1,\\
  a_{t+1}^{(r)},
    &y>a_{t+1}^{(r)}+\Delta_t^{(r)}-1,
 \end{cases}
\]
with $\tau_r(0)=0$.  The final comparison with $\Delta_t^{(r)}$ is handled by the same cap: if $2^{q_r(t)}\ge b$, then $\Delta_t^{(r)}>y$, while otherwise the summand has at most $b$ bits.  Hence $\tau_r$ is effective in polynomial time and $O(b)$ space.

For $t\ge1$, monotonicity of the gaps gives
\[
 \Delta_{t-1}^{(r)}
 \le a_t^{(r)}
 \le t\,\Delta_{t-1}^{(r)}.
\]
Since
\[
 q_r(t-1)=\Theta((\log_2 t)^{1/r}),
\]
it follows that
\[
 \log_2 a_t^{(r)}
 =2^{\Theta((\log_2 t)^{1/r})}.
\]
By the threshold law, $\delta_{\tau_r}(u)$ is the least $t$ with $2^u-1\le a_t^{(r)}$.  Inverting the preceding estimate yields
\[
 \delta_{\tau_r}(u)=2^{\Theta((\log_2 u)^r)}.
\]

To verify profile closure, choose constants $c_r,C_r>0$ such that for all sufficiently large $u$,
\[
 2^{c_r(\log u)^r}
 \le \delta_{\tau_r}(u)
 \le 2^{C_r(\log u)^r}.
\]
Given $u\in\mathcal W_0$ and fixed $d,C$, choose a fixed integer $m$ so large that $c_rm^r>C_r+d+1$, and let $v$ be a sufficiently large fixed multiple of $u^m$, still in $\mathcal W_0$.  Since $\log v=m\log u+O(1)$ and $\log u\le(\log u)^r$ for large $u$,
\[
 \delta_{\tau_r}(v)
 \ge 2^{c_r(\log v)^r}
 \ge C\,2^{C_r(\log u)^r}u^d
 \ge C\,\delta_{\tau_r}(u)u^d
\]
for all sufficiently large $N$, while $v\ge Cu$.  Hence $\mathcal W_0$ is $\tau_r$-profile-closed.

For $r<s$, the same width gives $\tau_r\preceq_{\mathcal W_0}\tau_s$, because $(\log u)^r=o((\log u)^s)$.  The reverse domination is impossible: take $u=\mathbf N$.  Every $v\in\mathcal W_0$ satisfies $v(N)=N^{O(1)}$ by Lemma~\ref{lem:widths}, and therefore
\[
 \delta_{\tau_r}(v)
 =2^{O((\log N)^r)}
 =o\!\left(2^{\Omega((\log N)^s)}\right)
 =o(\delta_{\tau_s}(\mathbf N)).
\]
Thus $\tau_r\prec_{\mathcal W_0}\tau_s$.  For $r=1$, the depth is polynomial in $u$, so power reparameterization in $\mathcal W_0$ gives $\tau_1\asymp_{\mathcal W_0}\rho_{1,l}$.  Finally, every fixed $r$ satisfies $2^{(\log u)^r}=2^{o(u)}$, so $\tau_r\preceq_{\mathcal W_0}\rho_{0,l}$; the reverse domination fails at $u=\mathbf N$ because every polynomially bounded $v$ has $\delta_{\tau_r}(v)=2^{o(N)}$, whereas $\delta_{\rho_{0,l}}(\mathbf N)=2^{\Theta(N)}$.
\end{proof}

The local primitives and branching semantics are fixed throughout the hierarchy; only descent varies.  Hence its strictness shows that descent resource cannot always be traded for polynomially more mutable state.

\begin{corollary}[Further quotient collapses]\label{cor:collapse}
For every fixed $k,l\ge1$,
\[
 \rho_k^{\rm sh}\asymp_{\mathcal W_0}\rho_{1,l},
 \qquad
 \rho_k^{\rm mid}\asymp_{\mathcal W_0}\rho_{0,l}.
\]
By Corollary~\ref{cor:compare}, the corresponding $\SR^Q$ classes are therefore equal for every Boolean branch mode $Q$.  For $k>1$, the distinct pointwise depth curves
\[
 u^{1/k},\quad u,\quad 2^{u/k},\quad 2^u
\]
therefore collapse in pairs in the descent quotient.  Together with Theorem~\ref{thm:quotient-hierarchy}, this shows that width reparameterization removes some apparent depth distinctions but not all of them: the quotient has both nontrivial identifications and infinitely many strict levels.
\end{corollary}

\begin{proof}
All four descents are profile-closed on $\mathcal W_0$: power cofinality handles the divisive and shallow polynomial factors, while scaling closure handles the predecessor and intermediate exponential factors.  For the shallow profile, $u^{1/k}=O(u)$ gives one domination, while $v=u^k\in\mathcal W_0$ gives $v^{1/k}=u$.  For the intermediate exponential profile, $2^{u/k}=O(2^u)$ gives one domination, while $v=ku\in\mathcal W_0$ gives $2^{v/k}=2^u$.  Apply Corollary~\ref{cor:compare} in both directions.
\end{proof}

\begin{lemma}[Canonical profile closure]\label{lem:profileclosure}
Fix $l\ge1$.
\begin{enumerate}[label=(\alph*),leftmargin=2em]
 \item Every scaling-closed, power-cofinal width family is $\rho_{1,l}$-profile-closed.
 \item Every scaling-closed width family is $\rho_{0,l}$-profile-closed.
\end{enumerate}
\end{lemma}

\begin{proof}
For $\rho_{1,l}$, Lemma~\ref{lem:duality} gives $\delta_{\rho_{1,l}}(u)=\Theta(u)$.  Power cofinality applied to $u^{d+1}$, followed by scaling closure, supplies a width term dominating both $Cu$ and $C\delta_{\rho_{1,l}}(u)u^d$.

For $\rho_{0,l}$, Lemma~\ref{lem:duality} gives $\delta_{\rho_{0,l}}(u)=\Theta(2^u)$.  Since every width term is at least $\mathbf L\ge2$, choosing by scaling closure a sufficiently large fixed multiple $v\ge m u$ makes $2^v$ dominate $C2^u u^d$ while also giving $v\ge Cu$.
\end{proof}

\begin{theorem}[Time--space representation]\label{thm:timespace}
Fix $l\ge1$ and $Q\in\{\D,\Ex,\All,\Alt\}$.
\begin{enumerate}[label=(\alph*),leftmargin=2em]
 \item If $\mathcal W$ is scaling-closed and power-cofinal, then
 \[
  \SR^Q_{1,l}[\mathcal W]=\mathsf{TIME}_Q(\mathcal W).
 \]
 \item If $\mathcal W$ is scaling-closed, then
 \[
  \SR^Q_{0,l}[\mathcal W]=\mathsf{SPACE}_Q(\mathcal W).
 \]
\end{enumerate}
\end{theorem}

\begin{proof}
For row $1$, Lemmas~\ref{lem:duality} and~\ref{lem:profileclosure}(a), together with Theorem~\ref{thm:exactprofile}, give
\[
 \SR^Q_{1,l}[\mathcal W]
 =\bigcup_{u\in\mathcal W}\TISP_Q(\Theta(u),u).
\]
A branch-time bound $O(u)$ already bounds work space by $O(u)$, so the right side is exactly $\mathsf{TIME}_Q(\mathcal W)$.

For row $0$, Lemmas~\ref{lem:duality} and~\ref{lem:profileclosure}(b) give the exact profile
\[
 \SR^Q_{0,l}[\mathcal W]
 =\bigcup_{u\in\mathcal W}\TISP_Q(2^u,u).
\]
This is contained in $\mathsf{SPACE}_Q(\mathcal W)$.  Conversely, let a total machine use $O(u)$ space.  Lemma~\ref{lem:total} bounds every branch by $2^{O(u)}$ steps.  Scaling closure supplies $v\in\mathcal W$ whose predecessor profile $\Theta(2^v)$ dominates that bound and whose width stores the configurations.  Hence the language lies in the displayed profile class and therefore in $\SR^Q_{0,l}[\mathcal W]$.
\end{proof}

\section{Classical calibrations}

The structural theorems recover the usual classes as calibration points of the resource geometry. These identities are consequences rather than the primary results.  For every fixed $l\ge1$,
\[
\begin{array}{c|cc}
Q&\rho_{1,l}\text{ on }\mathcal W_0&\rho_{0,l}\text{ on }\mathcal W_0\\ \hline
\D&\Pclass&\PSPACE\\
\Ex&\NP&\PSPACE\\
\All&\co\NP&\PSPACE\\
\Alt&\PSPACE&\EXP
\end{array}
\]
and on logarithmic width the predecessor row gives $\Lclass,\NL,\co\NL,\Pclass$ for $Q=\D,\Ex,\All,\Alt$, respectively.  These follow from Theorem~\ref{thm:timespace}, $\NL=\co\NL$, Savitch's theorem, and the standard alternation identities \cite{Immerman1988,Szelepcsenyi1988,Savitch1970,CKS1981}.  For $k\ge1$, let $\Sigma_k^P$ and $\Pi_k^P$ denote polynomial-time alternating languages with at most $k-1$ alternation switches, beginning existentially and universally, respectively.  Restricting the divisive alternating row accordingly yields these classes because both simulations preserve node types and alternation count \cite{Oitavem2022}.

The same statement at higher width scales gives
\[
\begin{array}{c|cc}
Q&\rho_{1,l}\text{ on }\mathcal W_k&\rho_{0,l}\text{ on }\mathcal W_k\\ \hline
\D&k\text{-}\mathsf{EXPTIME}&k\text{-}\mathsf{EXPSPACE}\\
\Ex&\mathsf N k\text{-}\mathsf{EXPTIME}&k\text{-}\mathsf{EXPSPACE}\\
\All&\mathsf{coN}k\text{-}\mathsf{EXPTIME}&k\text{-}\mathsf{EXPSPACE}\\
\Alt&k\text{-}\mathsf{EXPSPACE}&(k+1)\text{-}\mathsf{EXPTIME}.
\end{array}
\]
For $\mathcal W_E$, the divisive deterministic, existential, and alternating rows are $\Eclass,\NE,\ESPACE$, and the predecessor deterministic, existential, and universal rows are $\ESPACE$.  All displayed calibrations follow from Lemma~\ref{lem:widths}, Savitch, and the standard alternation simulations $\ATIME(t)\subseteq\DSPACE(t)$, $\DSPACE(s)\subseteq\ATIME(s^{O(1)})$, and $\ASPACE(s)=\DTIME(2^{O(s)})$ at the relevant constructible scales. Polynomial distortions are absorbed by the width families \cite{CKS1981}.

\subsection{Counting}

Counting is a further aggregation extension of the same local dynamics, not part of the Boolean profile-equivalence theorem above.  It replaces the existential Boolean aggregator by addition,
\[
 C(x,c,y)=
 \begin{cases}
  a(x,c),&h(x,c)=1,\\
  0,&h(x,c)=0\text{ and }y=0,\\
  \displaystyle\sum_{i:\eta_i(x,c)=1}C(x,T_i(x,c),\rho(y)),&\text{otherwise}.
 \end{cases}
\]
Let $\SR^{\Cnt}_{\rho}[\mathcal W]$ be the resulting integer-valued functions.  Define $\#\Pclass$ as the functions counting accepting branches of a nondeterministic polynomial-time machine.

\begin{corollary}[Counting]
For every fixed $l\ge1$,
\[
 \SR^{\Cnt}_{1,l}[\mathcal W_0]=\#\Pclass.
\]
\end{corollary}

\begin{proof}
Polynomial width and the divisive descent give polynomial branch depth.  A simulator chooses one of the fixed successors and rejects disabled choices, so the simulation is parsimonious: accepting branches are preserved one-for-one. Fixed $q$-ary branching is converted to binary branching with unused codes rejecting.  Conversely, Lemma~\ref{lem:config} compiles any fixed nondeterministic polynomial-time machine and the counting recursion reproduces its accepting paths.  This is exactly Valiant's definition \cite{Valiant1979}.
\end{proof}

\section{Conclusion}

Step recursion admits a machine-resource factorization once state width and recursion descent are separated.  The deterministic bounded-state core is realized by a single ordinary bounded step recursion over a fixed finite basis, while deterministic, existential, universal, and alternating aggregation of the same local dynamics gives the corresponding finite-branching machine semantics.  After controlled width closure, the resulting classes are the machine time--space families determined by the depth--width profiles $(\delta_\rho(u),u)$.

For generalized-inverse descents, generator growth determines depth.  The principal structural output is the quotient under admissible width change: some depth curves collapse, while an explicit infinite strict hierarchy survives every polynomial width reparameterization between the canonical polynomial- and exponential-depth profiles.  Thus descent remains a nonredundant resource coordinate after polynomial storage overhead is factored out; standard complexity classes serve only as calibration points.

The algebra--machine bridge is deliberately limited to the bounded-state core and does not identify the full closure $\mathcal A_\rho(\mathcal C_{\rm bit})$ with a machine class.  This distinction keeps the resource theorem exact: width remains an explicit storage parameter, descent remains an explicit path-length parameter, and the local transition basis remains fixed.

\section*{Declaration of competing interest}
The author declares that he has no known competing financial interests or personal relationships that could have appeared to influence the work reported in this paper.

\section*{Data availability}
No data were used for the research described in this article.

\section*{Declaration of generative AI and AI-assisted technologies in the manuscript preparation process}
During the preparation of this work, the author used OpenAI ChatGPT for copyediting, grammatical correction, language polishing. The author reviewed and edited the output as needed and takes full responsibility for the content of the article.

\fontsize{7}{7.5}\selectfont

\end{document}